\documentclass{article}
\usepackage{amsmath,amsthm,amsfonts,graphicx,amssymb}
\usepackage{multirow}
\usepackage{physics}
\usepackage{subcaption}
\usepackage{booktabs}
\usepackage{algorithm}
\usepackage[noend]{algpseudocode}
\usepackage{enumerate}

\newtheorem{theorem}{Theorem}
\newtheorem{lemma}{Lemma}

\newtheorem{definition}{Definition}
\newtheorem{proposition}{Proposition}
\newtheorem{corollary}{Corollary}

\newtheorem{fact}{Fact}
{\theoremstyle{definition}
}

\usepackage{color}
\usepackage[export]{adjustbox}
\usepackage{cleveref}

\usepackage{tikz}
\usepackage{scalerel}
\usetikzlibrary{matrix,calc,shapes,bending} 
\usetikzlibrary{knots,calc}

\usetikzlibrary{decorations.pathmorphing} 
\usetikzlibrary{fit}					
\usetikzlibrary{backgrounds}	
\usetikzlibrary{fit, positioning}

\usepackage{bold-extra}

\usepackage{fullpage}
\usepackage{url}

\DeclareMathSymbol{\mlq}{\mathord}{operators}{``}
\DeclareMathSymbol{\mrq}{\mathord}{operators}{`'}

\makeatletter
\newcommand\niton{\mathrel{\m@th\mathpalette\canc@l\owns}}
\newcommand\canc@l[2]{{\ooalign{$\hfil#1/\mkern1mu\hfil$\crcr$#1#2$}}}
\makeatother

\newif\ifcccg

\cccgfalse

\begin{document}

\title{Red-Blue-Partitioned MST, TSP, and Matching\footnote{A preliminary version of this work was presented at CCCG 2018 \cite{johnsonredbluecccg2018}.}}

\author{Matthew P. Johnson\thanks{Department of Computer Science, Lehman College and Ph.D. Program in Computer Science, The Graduate Center, City University of New York}
}

\maketitle

\begin{abstract}
Arkin et al.~\cite{ArkinBCCJKMM17} recently introduced \textit{partitioned pairs} network optimization problems: given a metric-weighted graph on $n$ pairs of nodes, the task is to color one node from each pair red and the other blue, and then to compute two separate \textit{network structures} or disjoint (node-covering) subgraphs of a specified sort, one on the graph induced by the red nodes and the other on the blue nodes. Three structures have been investigated by \cite{ArkinBCCJKMM17}---\textit{spanning trees}, \textit{traveling salesperson tours}, and \textit{perfect matchings}---and the three objectives to optimize for when computing such pairs of structures: \textit{min-sum}, \textit{min-max}, and \textit{bottleneck}. We provide improved approximation guarantees and/or strengthened hardness results for these nine NP-hard problem settings.

\end{abstract}


\newcommand{\StFactPar}{3}
\newcommand{\StFact}{3}
\newcommand{\StFactMo}{2}
\newcommand{\StFactBf}{\mathbf{3}}

\newcommand{\mmStFact}{4}
\newcommand{\mmStFactBf}{\mathbf{\mmStFact}}

\newcommand{\msTspFact}{4}
\newcommand{\msTspFactBf}{\mathbf{\msTspFact}}

\newcommand{\mmTspFact}{4}
\newcommand{\mmTspFactBf}{\mathbf{\mmTspFact}}

\newcommand{\St}{\rho_\text{St}}
\newcommand{\tsp}{\rho_\text{tsp}}

\newcommand{\Stbf}{\rho_\mathbf{St}}
\newcommand{\tspbf}{\rho_\mathbf{tsp}}

\newcommand{\C}{$C_{6\times}$}

\newcommand{\Ccnsp}{\C-cover}

\newcommand{\Cc}{\Ccnsp\ }

\newcommand{\T}{\mathcal{T}}
\newcommand{\V}{\mathcal{V}}
\newcommand{\E}{\mathcal{E}}

\newcommand{\n}{\mathfrak{n}}

\hyphenation{sub-hypergraph}

\section{Introduction}

We consider the class of \textit{partitioned pairs} network optimization problems recently introduced by Arkin et al.~\cite{ArkinBCCJKMM17}.
Given a complete metric-weighed graph $G$ whose vertex set consists of $n$ pairs $\{p_1,q_1\},...,\{p_n,q_n\}$ (with $n$ even), the task is to color one node from each pair red and the other blue, and then to compute two \textit{network structures} or disjoint (node-covering) subgraphs of a specified sort, one on the graph induced by the blue nodes and the other on the red nodes. One motivation is robustness: if the pairs represent $n$ different types of resources needed to build the desired network structure, with two available instances $p_i,q_i$ of each type $i$, then solving the problem means computing two separate independent instances of the desired structure, one of which can be used as a backup if the other fails. 

The structures that have been investigated are \textit{spanning trees}, \textit{traveling salesperson}, and \textit{perfect matchings}. A solution consists of a disjoint pair of subgraphs covering all nodes, i.e., two (partial) matchings, two trees, or two cycles, and there are different potential ways of evaluating the cost of the pair. The optimization objectives that have been considered are: 1) minimize the sum of the two structures' costs ({\em min-sum}), 2) minimize the maximum of the two structures' costs ({\em min-max}),  and 3) minimize the weight of the heaviest edge used in either of the structures ({\em bottleneck}).

\begin{table*}[htbp]
\caption{
Summary of results. $R, B \subseteq E$ denote the red and blue solutions, respectively. UB values indicate the approximation factors we obtain, {\em all for general metric spaces}; LB values indicate hardness of approximation lower bounds, {\em all (except min-sum and min-max TSP) for the special case of metric weights $\{1,2\}$}. 
Best prior bounds (all due to \cite{ArkinBCCJKMM17}) are also shown, where $\rho_\text{St} \le 2$ denotes the underlying metric space's Steiner ratio (conjectured to be ${2 \over \sqrt{3}} \approx 1.1547$ in Euc.~2D \cite{ivanov2012steiner}),
and $\tsp$ denotes TSP's best achievable approximation factor in the underlying metric space (currently $\tsp=1.5$ in general \cite{christofides1976worst}).
\label{tbl:results}
}
\vskip -.5cm
\centering
{ 
\ifcccg
\newcommand{\smaller}{\footnotesize}
\small
\begin{tabular}{cr@{\hskip .75cm}cc@{\hskip -.05cm}c} \\
\else
\newcommand{\smaller}{\small}
\begin{tabular}{cr@{\hskip 1cm}cc@{\hskip -.025cm}c} \\
\fi
\toprule
&&  min-sum & min-max & bottleneck \\
&&  $c(R)+c(B)$ & $\max\{c(R),c(B)\}$ & $\max\{w_e : e \in B \cup R\}$ \\
\midrule

\multirow{4}{*}{MST}
 & \textbf{{our UB}}: & \normalsize{$\StFactBf$} & \normalsize{$\mmStFactBf$} & \normalsize{$\mathbf{-}$} \\
 & \smaller{\cite{ArkinBCCJKMM17}'s UB}: & $\smaller{(3 \St)}$ & $\smaller{(4 \St)}$ & $\smaller{(9)}$ \\
 & \textbf{{our LB}}: & \textbf{NP-h} & \textbf{NP-h} & \textbf{2} \\
 & \smaller{\cite{ArkinBCCJKMM17}'s LB}: & (-) & \smaller{(NP-h in metric)} & (-) \\
 \cmidrule{2-5}
\multirow{4}{*}{TSP}
 & \textbf{{our UB}}: & \normalsize{$\msTspFactBf$} & \normalsize{$\mmTspFactBf$} & \normalsize{$\mathbf{-}$} \\
 & \smaller{\cite{ArkinBCCJKMM17}'s UB}: & $\smaller{(3 \tsp)}$ & $\smaller{(6 \tsp)}$ & $\smaller{(18)}$ \\
\ifcccg
\else
 & \textbf{{our LB}}: & \multicolumn{2}{c}{\hskip -.7cm \normalsize{$\mathbf{123/122 \approx 1.00819}$} \smaller\textbf{\em with metric weights $\mathbf{\{.5,1,1.5,2\}}$}} & \normalsize{$\mathbf{2}$}\\
\fi
 & \smaller{\cite{ArkinBCCJKMM17}'s LB}: & (-) & \smaller{(-)} & (-) \\
 \cmidrule{2-5}

\multirow{4}{*}{matching}
 & \textbf{{our UB}}: & \normalsize{$\mathbf{-}$} & \normalsize{$\mathbf{-}$} & \normalsize{$\mathbf{-}$} \\
 & \smaller{\cite{ArkinBCCJKMM17}'s UB}: & $\smaller{(2)}$ & $\smaller{(3)}$ & $\smaller{(3)}$ \\
 & \textbf{{our LB}}: & \normalsize{$\mathbf{{8305 \over 8304} \approx 1.00012}$} & \normalsize{$\mathbf{{8305 \over 8304} \approx 1.00012}$}  & \normalsize{$\textbf{2}$} \\
 & \smaller{\cite{ArkinBCCJKMM17}'s LB}: & \smaller{(NP-h in metric)} & \smaller{(weakly NP-h in 2D Euc.)}  & (-) \\

\bottomrule
\end{tabular}
}
\end{table*}

\noindent \textbf{Contributions.}
We provide a variety of results for these nine problem settings (all of which turn out to be NP-hard; see Table \ref{tbl:results}), including algorithms with improved approximation guarantees and/or stronger hardness results for each. In particular, we provide tighter analyses of the approximation factors 
of Arkin et al.~\cite{ArkinBCCJKMM17}'s min-sum/min-max 2-MST algorithm, which is equivalent to Algorithm \ref{alg:oumst} below. We show that the algorithm provides approximation guarantees of \StFact\ and \mmStFact\ for 2-MST with objectives min-sum and min-max, respectively.
%
We also show that a simple extension of this algorithm (see Algorithm~\ref{alg:outsp} below) provides  a 4-approximation for 2-TSP for both min-sum and min-max. All four approximation factors are tight.


\vskip .1cm
\noindent \textbf{Related work.}
The primary antecedent of this work is Arkin et al.~\cite{ArkinBCCJKMM17} (see also references therein), which introduced the class of 2-partitioned network optimization problems.
Earlier related problem settings include optimizing a path visiting at most one node from each pair \cite{gabow1976two}, generalized MST 
\ifcccg
\else
\cite{myung1995generalized,pop2004new,pop2001relaxation,bhattacharya2015approximation},
\fi
generalized TSP \cite{bhattacharya2015approximation},
%
constrained forest problems \cite{goemans1995general},
%
%
adding conflict constraints to MST \cite{zhang2011minimum,kante2013trees,darmann2011paths}
and 
to perfect matching \cite{oncan2013minimum,darmann2011paths}, and
balanced partition of MSTs \cite{andersson2003balanced}.

%
%
%
%


\section{2-MST}

\ifcccg
\else
\subsection{Decomposing a 2-colored spanning tree}
In this section we prove a key lemma used in the next section's approximation analysis, on the result of partitioning a metric-edge-weighted spanning tree into a 2-component spanning forest. Specifically, we show that {\em for any 2-coloring} $V_b \cup V_r = V$ of an arbitrary metric-edge-weighted graph ({\em even without the constraint of specified pairs being colored differently}), the sum of the costs of MSTs on $V_b$ and $V_r$ will be at most three times an MST on $V$, and each of them alone will be at most two times this. 
\fi
\ifcccg

\else
\fi

\begin{lemma}\label{lem:2submsts}
Let $V$ be the nodes of a metric-weighted graph.
Let $T$ be an MST on $V$, and let $V_{b} \cup V_{r} = V$ be any 2-coloring of $V$, and let $T_{b}$ and $T_{r}$ be MSTs of $V_{b}$ and $V_{r}$, respectively. Then we have:
\begin{enumerate}[(a)]
\item $c(T_{b})+c(T_{r}) ~\le~ \StFact c(T)$, \:and
\item $\max\{c(T_{b}),c(T_{r})\} ~\le~ \StFactMo c(T)$.
\end{enumerate}
\end{lemma}
\ifcccg
\else
\begin{proof}
Pick an arbitrary node $v^*$ as the root, 
and, for the purposes of this proof, impose an orientation on all edges as directed away from the root.

First consider a monochromatic component $H$ of the graph, i.e., one of the components that would be produced by deleting all bichromatic edges. That is, all the nodes of $H$ are the same color, say, blue. Let $H$'s {\em root} be its node closest to $v^*$ (if $H$ does not contain $v^*$), let its {\em parent} be its (red) neighbor that is the next node on the path to the root, and let its {\em children} be its other neighbors (also red).

Now let us calculate the cost paid by $T_b$ and $T_r$ for $H$'s internal edges. Clearly $T_b$ pays once for each internal edge, i.e., $c(H)$. $T_r$ may wish to visit each of $H$'s nodes (in order, e.g., to reach red neighbors of them), but traversing an edge-doubled $H$ will not cost $T_r$ more than $2c(H)$ (see Fig.~\ref{fig:monochromcompstraight}).

Therefore if we shrunk each monochromatic component $H$ to a single node (see Fig.~\ref{fig:monochromcompcurved}), charging $c(H)$ to $T_b$ and $2c(H)$ to $T_r$ when $H$ is blue and the reverse when $H$ is red, this would render all remaining edges of the resulting shrunken graph bichromaric, and it would pay for $T_b,T_r$ to both reach all nodes within $H$.
Moreover, consider an MST $\hat T_\chi$ of the color-$\chi$ nodes in the shrunken graph. Observe that the tree that $\hat T_\chi$ would induce in the original graph is exactly $T_\chi$, and that the edges $\hat T_\chi - T_\chi$ are exactly the monochromatic edges that $T_\chi$ was already charged (once or twice apiece) for. Therefore assume for simplicity henceforth that all monochromatic components are single nodes,  i.e., all edges of the graph are bichromatic.

Let the {\em depth} of a node in the shrunken graph be the number of hops in its path to $v^*$ (where in the shrunken graph $v^*$ now refers to the node representing the monochromatic component containing $v^*$ in the original graph). 

Root $v^*$ has some color, say, red. Then notice that all red nodes will have even depth and all blue nodes odd depth.

Now, one way $T_r$ could connect a blue node $v_b$'s red parent to its red children is by following a path from $v_b$'s parent to one of its children, and then visiting the in sequence (see Fig.~\ref{fig:monochromcompshrunkcurved}). Then $T_r$ pays for $v_b$'s child edges at most twice and for its {\em red} parent edge only once. For each red node with red grandchildren, connect them thus.
Similarly, by constructing analogously the portion of $T_b$ appearing one level down, beginning with an outgoing edge from $v_b$, we can ensure that $T_b$ pays only once for $v_b$'s child edges, although it could potentially pay twice for $v_b$'s parent edge. For each blue node with blue grandchildren, connect them thus. Finally, connect the (blue) children of the root together sequentially.



Now, first consider $T_r$. Observe that $T_r$ will pay once for each of $v^*$'s edges, shortcutting between each successive pair, and that more generally, $T_r$ will pay only once for every edge {\em from} a red (even-depth) node {\em to} a blue (odd-depth) node. The only type of place where shortcutting will {\em not} be possible, where $T_r$ will potentially pay twice for edges, will be edges {\em from} a blue (odd-depth) nodes {\em to} red (even-depth) nodes. (In essence, $T_r$ will shortcut from red grandparent to red grandchild, and will then traverse the doubled edges from one red grandchild to the next.)
Symmetrically, the charges to $T_b$ are exactly the opposite of this, paying once for odd-to-even-depth edges and potentially twice for even-to-odd-depth edges.

Thus every edge is paid for at most twice by each of $T_b,T_r$, and at most thrice in total.
%
%
\end{proof}

\fi

\ifcccg
\else
Now we refine the argument to improve the combined cost of the two trees slightly, reducing it by the weight of three heavy edges in the following result lemma, which will be a key lemma in proving approximation ratios for the 2-MST problem. %

\begin{theorem}\label{thm:2submsts}
Let $V$ be the nodes of a metric-weighted graph.
Let $T$ be an MST on $V$. Let $V_{b} \cup V_{r} = V$ be any 2-coloring of $V$, and let $T_{b}$ and $T_{2}$ be MSTs of $V_{b}$ and $V_{r}$, respectively. 
Let $e_\times=\{v_L,v_R\}$ be a heaviest edge in $T$, with weight $w_\times$. Let $T_L,T_R$ be the trees (on nodes $V_L,V_R$, respectively) obtained by deleting $e_\times$ from $T$, where $v_L \in T_L$ and $v_R \in T_R$. Let $w_L,w_R$ be the heaviest edge weights appearing in $T_L,T_R$, respectively.
Then we have:
\begin{enumerate}[(a)]
\item $c(T_{b})+c(T_{r}) ~\le~ \StFact c(T)-(w_L+w_\times+w_R)$, \:and
\item $\max\{c(T_{b}),c(T_{r})\} ~\le~ \StFactMo c(T)-(w_L+w_\times+w_R)$.
\end{enumerate}
Moreover, if all nodes of $T_L$ are, say, blue, then:
\begin{enumerate}[(a)]
\setcounter{enumi}{2}
\item $c(T_{b})+c(T_{r}) ~\le~ c(T_L) + w_\times + (\StFact c(T_R)-w_R)$, \:and
\item $\max\{c(T_{b}),c(T_{r})\} ~\le~ c(T_L) + w_\times + (\StFactMo c(T_R) - w_R)$.
\end{enumerate}
\end{theorem}
\fi

\ifcccg
\else

\begin{figure}[t!]
\center
    \begin{subfigure}[t]{5.2cm}
        \centering
            \includegraphics[width=4cm]{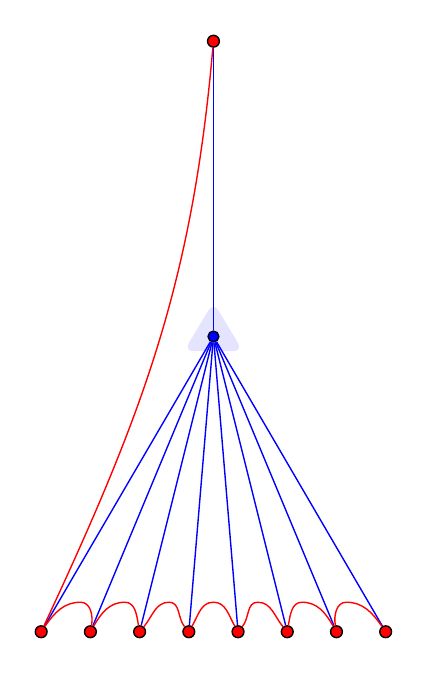}
\vskip -.075cm
\caption{A subpath of $T_r$ passing through a blue node (representing $H$).}
        \label{fig:monochromcompshrunkcurved}
    \end{subfigure}
\hskip 0.35125cm
    \begin{subfigure}[t]{5.2cm}
        \centering
            \includegraphics[width=4cm]{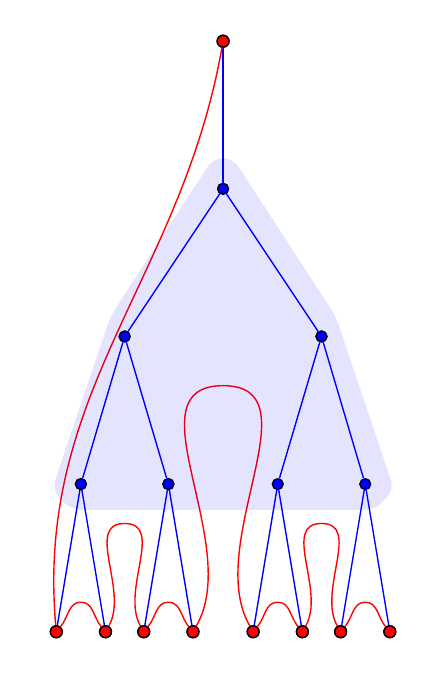}
\vskip -.075cm
\caption{The blue node expanded to the underlying $H$.}
        \label{fig:monochromcompcurved}
    \end{subfigure}
\hskip 0.35125cm
    \begin{subfigure}[t]{5.2cm}
        \centering
            \includegraphics[width=4cm]{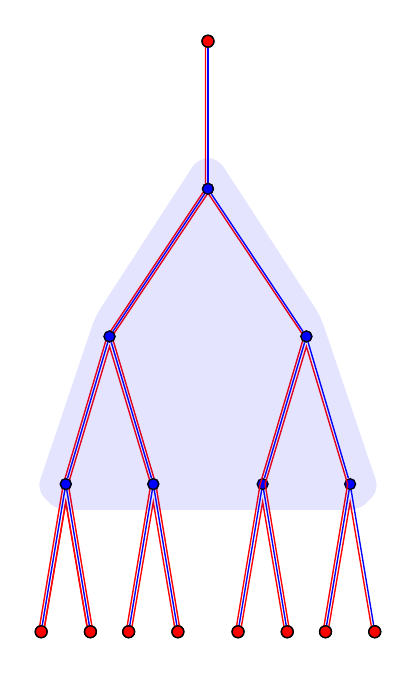}
\vskip -.075cm
\caption{The path's edges expanded to subpaths in the underlying graph.}
        \label{fig:monochromcompstraight}
    \end{subfigure}
\caption{Example portion of $T_r$ visiting a monochrome (blue) component $H$ (shaded) and edges incident to $H$. $T_r$ twice once for all edges except those on the path from the upper red node to the lower right node.}\label{fig:monochromcomp}
\end{figure}

\begin{proof}
%
%
%
((a) and (b).)
If each of the walks with shortcuts performed in the proof of Lemma~\ref{lem:2submsts} were {\em closed} walks, returning to their starting nodes, then each edge would be paid for three times. But the walks do not need to return to their starting nodes, they only have to visit all nodes. Therefore we can (among other more complicated options) choose two leaves as start and end nodes and pay only {\em twice} for the edges on the path between them.\footnote{In fact, we can avoid triple payment of additional edges, within and potentially incident to {\em every} shrunken node.}
 For a shrunken node $v_H$ resulting from a monochromatic component $H$ (of color, say, blue, and whose ``root'' is its node closest to $v^*$), let the {\em leafed monochromatic component} $\hat H$ be the union of $H$ and any outgoing edges incident to $H$, i.e., edges connecting $H$ to red nodes, {\em except} (if $H$ does not include $v^*$) for the edge incident to $H$'s root on the path to $v^*$. Then we can avoid $T_r$'s double payment on the edges on any chosen root-to-leaf path within $\hat H$ (see Fig.~\ref{fig:monochromcomp});
In particular, we can choose two leaf nodes whose path includes $e_L,$, $e_\times$, and $e_R$, thus avoiding the third charge for those edges. This strategy will yield {\em a} tree spanning blue nodes and {\em a} tree spanning red nodes satisfying inequalities (a) and (b). Therefore MSTs of the blue nodes and the red nodes, respectively, will satisfy them as well.

((c) and (d).) Again choose two leaves as start and end nodes of a path, this time including $e_\times$ and $e_R$. All edges on the portion of this path within $T_R$ will be paid for only twice, but because $T_L$ is monochromatic, $e_\times$ and all edges within $T_L$ will already be paid for only once.
\end{proof}

\fi

\ifcccg
\else
\begin{figure*}[t!]
\center
    \begin{subfigure}[t]{8.028388cm}
        \centering
            \includegraphics[width=7.5cm]{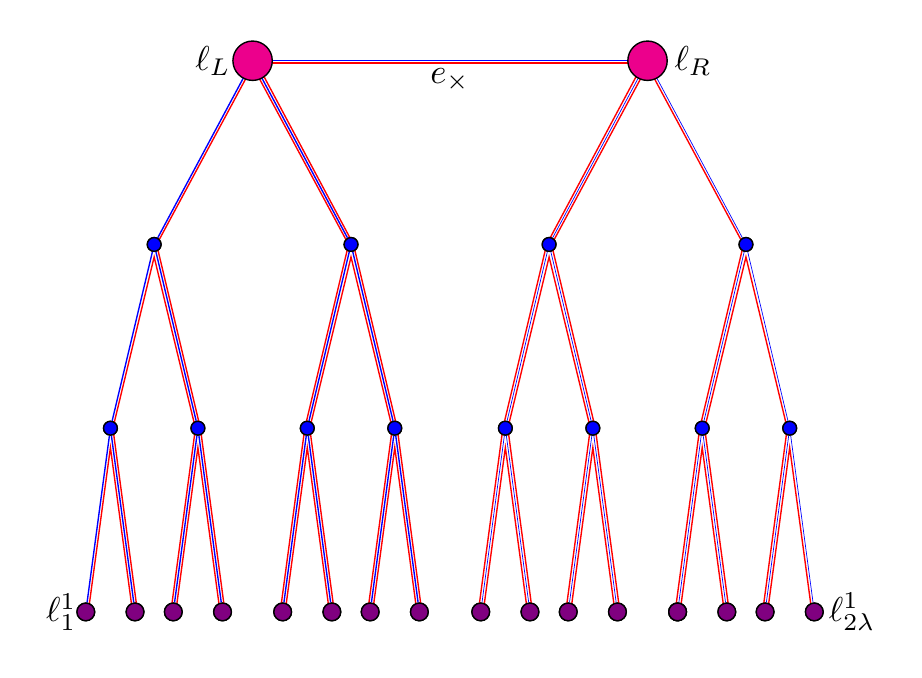}
\vskip -.2cm
            \caption{Coloring with one node red and one blue at every leaf point, and all nodes at non-leaf/non-$\ell_L/\ell_R$ points blue (and thus $(\n+1)/2$ nodes red and $\n-1$ blue at each of $\ell_L$ and $\ell_R$). Results in edges on the path from $\ell_1^1$ to $\ell_{2\lambda}^1$ paid for once by $ALG_b$ and once by $ALG_r$, and all others paid for once by $ALG_b$ and {\em twice} by $ALG_r$, and so $c(ALG_b) = c(T)$ and $c(ALG_r) \approx 2c(T)$.}
        \label{fig:tightmstbad}
    \end{subfigure}
\hskip .35cm
    \begin{subfigure}[t]{8.028388cm}
        \centering
            \includegraphics[width=7.5cm]{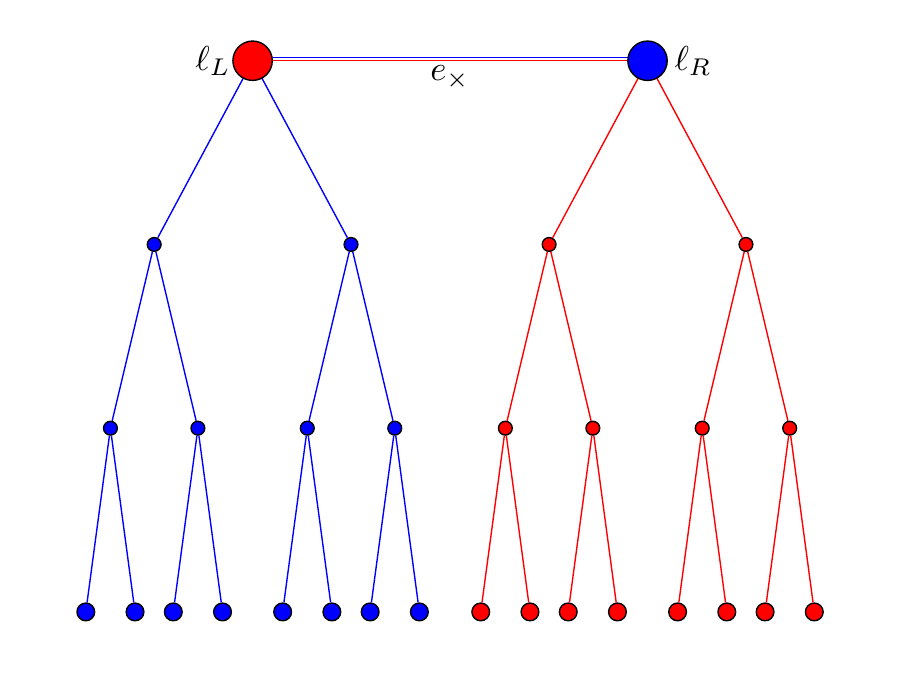}
\vskip -.2cm
\caption{Coloring with all nodes $\ell_L$ red and all at $\ell_R$ blue (and thus all other nodes in $V_L$ blue and all others in $V_R$ red). Results in $e_\times$ paid for once by $OPT_b$ and once by $OPT_r$, and all other edges paid for once total, and so $c(OPT_b) = c(OPT_r) \approx c(T)/2$.}
        \label{fig:tightmstgood}
    \end{subfigure}
\caption{2-MST instance achieving Algorithm~\ref{alg:oumst}'s approximation factor 3 for min-sum and 4 for min-max (and the tightness of Lemma~\ref{lem:2submsts}'s inequalities (a) and (b)), drawn with two colorings. Its $2n=2(3\n-1)$ {\em nodes} are {\em (co-)located} at the $2\n$ {\em points} shown.
%
%
Two nodes are co-located at each of the $2\lambda=\n+1$ leaf points, one at each non-leaf/non-$\ell_L/\ell_R$ point, and $3\lambda-2=(3\n-1)/2$ at each of $\ell_L$ and $\ell_R$. Each node pair has one node at $\ell_L$ (or $\ell_R$) and one node at a descendent point of $\ell_L$ in $T_L$ (respectively, of $\ell_R$ in $T_R$).}\label{fig:tightmstgoodbad}
\end{figure*}
\fi

\begin{proposition}\label{obs:2submststight}
There exist families of graphs showing that bounds (a) and (b) of Lemma~\ref{lem:2submsts} are (simultaneously) tight.
\end{proposition}
\ifcccg
\else
\begin{proof}
We construct a graph as follows. First consider a set of $2\n$ points $\V=\V_L \cup \V_R$ in a metric space, arranged in the form of two full binary trees $\T_L,\T_R$ (with root points $\ell_L,\ell_R$, respectively, and each with $\lambda$ leaf points and $\n=2\lambda-1$ node points overall; see Fig.~\ref{fig:tightmstgoodbad}), in the sense that $c(\{\ell,\ell'\})=1$ for every edge $\{\ell,\ell'\} \in \T_L \cup \T_R$. Let $c(e_\times)=1+\epsilon$, where $e_\times=\{\ell_L,\ell_R\}$, and let $\T = \T_L \cup \{e_\times\} \cup \T_R$. Set distances between all other pairs of points of $\V$ equal their path distances in $\T$.

Now we define a metric-weighted graph on $2(3\n-1)$ nodes $V$, which are located at points of $\V$ as follows. Two nodes are co-located at each of $\T$'s $2\lambda=\n+1$ leaf points; one node is located at each of $\T$'s $2\lambda-4=\n-3$ non-leaf/non-root points; finally, $3\lambda-2=(3\n-1)/2$ nodes are co-located at each of $\ell_L$ and $\ell_R$. An MST $T$ on $V$ will pay (by construction) once for each edge of $\T$, for a total cost (ignoring the additional $\epsilon$ in $c(e_\times)$) of $4\lambda-2 = 2\n-1$.



Consider the following coloring: at each leaf point one node is red and one is blue, each non-leaf/non-root point's node is blue, and half of $\ell_L$'s and $\ell_R$'s nodes are red and half are blue.



Then an MST $ALG_b$ of the blue nodes will pay once for each edge of $\T$, totaling $2\n-1$, and an MST $ALG_r$ of the red nodes will pay {\em twice} for each edge of $\T$, except for those on some longest path between nodes at leaf pointss, say, from $\ell_1^1$ to $\ell_{2\lambda}^1$, each of which it will only pay for once, and so
\begin{equation}\label{eq:algbrbounds}
c(ALG_b) = c(T)\text{,  ~and~  } c(ALG_r) = 2c(T)-\Theta(\log \n) \approx 2c(T).
\end{equation}
Thus we conclude: ${c(ALG_b)+c(ALG_r) \over c(T)} \to 3$ and ${\max\{c(ALG_b),c(ALG_r)\} \over c(T)} \to 2$.
\end{proof}

\fi

\ifcccg
\else
\subsection{Min-sum/min-max 2-MST: algorithm}
Now
\fi
we analyze 
Algorithm \ref{alg:oumst},
which forms trees $T_L,T_R$ by deleting a max-weight edge $e_\times$ (of weight $w_\times$) from an MST $T$ computed on the $2n$ nodes, and then colors all ``lone'' nodes appearing without their partners in $T_L$ blue and all lone nodes in $T_R$ red, and assigns arbitrary distinct colors to all other node pairs.


\ifcccg
\else

The proof analyzes three cases, depending on whether one, both, or neither $T_L,T_R$ contains a pair, the first two cases of which imply that $OPT$ must cross between $T_L$ and $T_R$ at least once or twice, respectively. The challenge is that $c(OPT)$ is lower-bounded by $c(T_L)+c(T_R)$ but not by  $c(T)=c(T_L)+w_\times+c(T_R)$.
We upper-bound $ALG$ by carefully applying Theorem~\ref{thm:2submsts} to $ALG_b+ALG_r$,
and we obtain a lower bound on $c(OPT)$ including $w_\times$ or $2w_\times$, permitting the two bounds to be compared, by subtracting max-weight edges from one or both sides.

The entities defined in the following definitions will be used throughout the rest of the subsection.
\begin{definition}
Let $V$ be a set of $n$ pairs of nodes $\{p_i,q_i\}$ of a metric-weighted graph,
and let $T$ be an MST on $V$.
%
%
Let $e_\times=\{v_L,v_R\}$ be a max-weight edge in $T$, with weight $w_\times$, and let $T_L,T_R$ be the trees (on nodes $V_L,V_R$, respectively) obtained by deleting $e_\times$ from $T$, where $v_L \in T_L$ and $v_R \in T_R$.
Let $e_L,e_R$ be max-weight edges (of weights $w_L,w_R$) be max-weight edges of $T_L$ and $T_R$, respectively. Let $T_L^- = T_L - \{w_L\}$ and $T_R^- = T_R - \{w_R\}$.
\end{definition}


\begin{definition}
Let $OPT$ be some particular optimal solution, and let $OPT_L$ and $OPT_R$ be the portions of $OPT$ induced by $V_L$ and $V_R$, respectively.
Let $ALG_L$ and $ALG_R$ be the portions of Algorithm \ref{alg:oumst}'s solution induced by $V_L$ and $V_R$, respectively.
\end{definition}

\begin{definition}
Say that any edge $e \in E$ is a {\em cross-edge} if $e$ has one node in $V_L$ and one node in $V_R$.
Say that a tree {\em contains a pair} if it contains both $p_i$ and $q_i$ for some $i$. Say that a node is a {\em lone node} if it lies in {\em one of} $T_L,T_R$, and its partner lies in {\em the other}.
\end{definition}

\begin{lemma}\label{lem:w}
If exactly one of $T_L,T_R$ contains a pair, then $OPT$ must contain a cross-edge of weight at least $e_\times$. If both $T_L$ and $T_R$ contain a pair, then $OPT$ must contain at least two cross-edges of weight at least $e_\times$, one in $OPT_L$ and one in $OPT_R$.
\end{lemma}
\begin{proof}
If only one of them contains a node pair, say, $V_L$ contains $\{p_L,q_L\}$, then $p_L$ and $p_R$ must be in different components of $OPT$. Thus $OPT$ must contain a path $P_L$ connecting $(p_L,...,T_R)$ or $(q_L,...,T_R)$, and so $P_L$ must contain a cross-edge of weight $\ge w_\times$.

If $V_L$ and $V_R$ each contain pairs, then the two members of each of these two pairs, say, $\{p_L,q_L\}$ and $\{p_R,q_R\}$, must be in different components of $OPT$ (which has exactly two components), $OPT$ must contain two {\em vertex-disjoint} paths $P_1$ and $P_2$ either connecting $(p_L,...,p_R)$ and $(q_L,...,q_R)$ or $(p_L,...,q_R)$ and $(q_L,...,p_R)$. Then each of these paths connects a node in $V_L$ to a node in $V_R$, and thus contains at least one cross-edge, both of weight $\ge w_\times$.
\end{proof}

Now we prove the approximation guarantee.
\fi

\begin{algorithm}[t!]
    \caption{Min-sum/min-max 2-MST approx}
    \label{alg:oumst}
    \begin{algorithmic}[1] 
            \State $T \gets$ an MST on the $2n$ nodes
            \State $\{T_L,T_R\} \gets$ result of deleting a max-weight edge $e_\times$ from $T$
	            \For {each node pair $(p_i,q_i) \in V_L \times V_R$}
	            	\Statex \hskip .52cm color $p_i$ blue and $q_i$ red
		\EndFor
	            \For {each other node pair $(p_i,q_i)$}
	            	\Statex \hskip .52cm assign $p_i,q_i$ arbitrary distinct colors
		\EndFor
	\ifcccg
\else
	            \For {$c \in \{b,r\}$}
	            	\Statex \hskip .52cm $T_c \gets$ an MST of the color-$c$ nodes
		\EndFor
	\fi
	\State \Return $\{T_b,T_r\}$
    \end{algorithmic}
\end{algorithm}

\begin{theorem}\label{thm:mstapxpf}
Algorithm \ref{alg:oumst} provides a $\StFact$-approximation for min-sum 2-MST.
\end{theorem}
%

\ifcccg
\else

\begin{proof}
We analyze three cases, depending on whether one, both, or neither of $T_L,T_R$ contains a 
pair.

\vskip .1cm
\noindent $\bullet$ (1) \textbf{Neither} $T_L$ nor $T_R$ contains a pair.
%
Then all nodes are lone nodes, then the solution is optimal. 

\vskip .1cm
\noindent $\bullet$ (2) \textbf{Both} $T_L$ and $T_R$ contain a pair.
%
%
%
%
%
%
Let $T_L^r$ and $T_L^b$ be MSTs on $T_L$'s red and blue nodes, respectively, and let $T_R^r$ and $T_R^b$ be MSTs on $T_R$'s red and blue nodes, respectively. 

If $v_L$ and $v_R$ are both the same color, say, blue, then edge $e_\times=(v_L,v_R)$ can be used to connect $T_L^b$ and $T_R^b$ with cost $w_\times$, but $e_\times$ cannot, by itself, be used to connect $T_L^r$ and $T_R^r$. What can be said, however, is that the cost of an edge between two red nodes $v_L^r \in V_L$ and $v_R^r \in V_R$ will be upper-bounded by the cost of the edges in {\em the path in $T_L$} from $v_L^r$ to $v_L$ plus $w_\times$ plus the cost of {\em the path in $T_R$} from $v_R$ to $v_R^r$. If $v_L$ and $v_R$ are different colors, say, blue and red, respectively, then $T_L^b$ and $T_R^b$ could be connected using $e_\times$ and a path to a blue node in $T_R$, and $T_L^r$ and $T_R^r$ could be connected using $e_\times$ and a path to a red node in $T_L$.

Now, consider the case of $T_L^r$ and $T_R^r$ when $v_L$ is blue. Consider $T_L$ as a tree rooted at $v_L$, and consider all the blue nodes in $V_L$ that are the {\em first blue nodes encountered} on {\em paths from $v_L$ in $T_L$}. If there is only one such node $T_L^b$ (there must be at least one), then in a solution connecting $T_L^b$ and $T_R^b$ using the path from $T_L^b$ to $v_L$, $ALG_b$ would only pay for this path once. ($ALG_r$ would also pay for it once.) If there are multiple such nodes, then a solution could be chosen in which $ALG_b$ pays twice for {\em all but one} of these paths, paying only once for that one. ($ALG_r$ would pay only once for all of them.)

Regardless of the location of edge $e_L$ within $T_L$, therefore, there will exist trees spanning $V^b$ and $V^r$ that together pay for $e_\times$ twice, pay for $e_L$ at most twice, and pay for all other edges at most thrice (and, similarly, that pay for $e_R$ at most twice).


%
%



Then the cost of the solution will be:
\begin{eqnarray}
c(ALG) &=& c(ALG_b) + c(ALG_r)\nonumber\\
&\le& \StFact c(T) - (w_L + w_\times + w_R)\\
&=& \StFact \Big( c(T_L) + c(T_R)) \Big) + (2 w_\times  - w_L - w_R),\label{eq:algbound}
\end{eqnarray}
where \eqref{eq:algbound} follows from Theorem~\ref{thm:2submsts}(a).

We know that
\begin{equation*}
c(T) - w_\times ~=~ c(T_L)+c(T_R) ~\le~ c(OPT),
\end{equation*}
and by Lemma~\ref{lem:w} we can assume both components $OPT_r,OPT_b$ of $OPT$ contain a cross-edge of weight at least $w_\times$. $OPT$ will be a 2-component spanning forest, and since both $T_L$ and $T_R$ contain a pair, $OPT_L$ and $OPT_R$ 
will each contain at least one fewer edge than $T_L$ and $T_R$, respectively.
First, suppose $OPT_L$ and $OPT_R$ each consist of two components, one blue and one red, i.e., forests with two trees, and $|V_L|-2$ and $|V_R|-2$ edges, respectively. 
But $T_L^-$ is a lightest-weight forest of $|V_L|-2$ edges on $V_L$ (hence $c(OPT_L) \ge c(T_L^-)$) and $T_R^-$ is a lightest-weight forest of $|V_R|-2$ edges on $V_R$ (hence $c(OPT_R) \ge c(T_R^-)$), and so:
%
\begin{eqnarray}
c(OPT) &\ge& c(OPT_L)+ 2w_\times + c(OPT_R) \nonumber\\
&\ge& c(T_L^-) + 2w_\times + c(T_R^-)\nonumber\\
&=& \Big(c(T_L)+c(T_R)\Big) + (2w_\times - w_L - w_R).\label{eq:optbound}
\end{eqnarray}

Second, suppose $OPT_L$ has two blue components rather than one. In this case $OPT_L$ consists of $|V_L|-3$ edges, of total cost at least $c(T_L^-)-w_L$, but now another cross-edge is required, having cost at least $w_\times$, which since $w_\times \ge w_L$ results in a net nonnegative increase in $c(OPT_L)$. More generally, additional components beyond two for either $OPT_L$ or $OPT_R$ would only increase lower bound \eqref{eq:optbound} further.

%

Combining \eqref{eq:optbound} and \eqref{eq:algbound}, we obtain:
\begin{eqnarray*}
{c(ALG) \over c(OPT)} 
&\le& {\StFact \Big( c(T_L) + c(T_R)) \Big) + (2 w_\times  - w_L - w_R) \over \Big(c(T_L)+c(T_R)\Big) + (2w_\times - w_L - w_R)}\nonumber\\
&\le& \StFact. \hskip 2.6cm \footnotesize{(\textit{because } w_\times \ge w_L,w_R)}
\label{eq:bothcontain}
\end{eqnarray*}


\noindent $\bullet$ (3) \textbf{Only} (say) $T_R$ contains a pair, with all nodes in $V_L$ being (say) blue.
Then all the red nodes (and some blues) lie in $V_R$. Then $c(ALG_b) = c(T_L^b)+c(T_R^b)+w_\times$ as before but now $c(ALG_r) = c(T_R^r)$. As in the discussion in case (2) above, regardless of the location of edge $e_R$ within $T_R$, there will exist a trees spanning $V^b$ and $V^r$ that together pay for $e_\times$ twice, pay for $e_R$ at most twice, and pay for all other edges at most thrice (but this time pay for $e_L$ at most only twice).


Then the cost of the solution will be:
%
\begin{eqnarray}
c(ALG) &=& c(ALG_b) + c(ALG_r)\nonumber\\
&\le& c(T_L) + w_\times + \Big(\StFact c(T_R) - w_R\Big)\nonumber\\
&=&  \Big(c(T_L) + \StFact c(T_R)\Big) + (w_\times - w_R),\label{eq:algbound1}
\end{eqnarray}
where \eqref{eq:algbound1} follows from Theorem~\ref{thm:2submsts}(c).

By Lemma~\ref{lem:w} we can assume that $OPT_b$ contains a cross-edge of weight at least $w_\times$. $OPT_L$ contains at least one (blue) component, and $OPT_R$ contains at least two components (one blue and one red). We can assume they contain exactly these many, since as above additional components would only increase the lower bound \eqref{eq:optbound1} further.
Because $T_L$ is an MST on $V_L$ (hence $c(OPT_L \ge c(T_L))$) and $T_R^-$ is a lightest-weight forest of $|V_R|-2$ edges on $V_R$ (hence $c(OPT_R) \ge c(T_R^-)$), we have:
\begin{eqnarray}
c(OPT) ~&\ge&~ c(OPT_L)  + w_\times + c(OPT_R)\nonumber\\
~&\ge&~ c(T_L) + w_\times + c(T_R^-)\nonumber\\
~&\ge&~  \Big(c(T_L)+c(T_R)\Big) + (w_\times - w_R).\label{eq:optbound1}
\end{eqnarray}
Combining \eqref{eq:algbound1} and \eqref{eq:optbound1} we obtain:
%
\begin{eqnarray*}
{c(ALG) \over c(OPT)} &\le& {\Big(c(T_L) + \StFactPar c(T_R) \Big) + (w_\times - w_R) \over \Big(c(T_L)+c(T_R)\Big) + (w_\times - w_R)}\\
&\le& \StFact. \hskip 2cm \footnotesize{(\textit{because } w_\times \ge w_R)}
\end{eqnarray*}
\vskip -.7cm
%
\end{proof}


\fi



This immediately implies that the same algorithm provides $6$-approximation for min-max 2-MST, but we perform a tighter analysis.
\begin{theorem}\label{cor:maxmstapxpf}
Algorithm \ref{alg:oumst} provides a $\mmStFact$-approximation for min-max 2-MST.
\end{theorem}


\ifcccg
\else
\begin{proof}
Let $OPT_{mm}$ be an optimal max-min solution, and let $c_{mm}(\cdot)$ denote the max-min 2-MST cost function. First, observe that
\begin{equation}\label{eq:mmmst1}
c_{mm}(OPT_{mm}) ~\ge~ c(OPT_{mm})/2 ~\ge~ c(OPT)/2.
\end{equation}
Therefore in the ``neither contains a pair'' case we obtain:
\begin{equation*}
c_{mm}(ALG) ~\le~ c(ALG) ~\le~ 2c_{mm}(OPT_{mm}).
\end{equation*}

In the ``both contain a pair'' case (we omit the ``only one contains a pair'' case, which is similar), we can, applying the second inequality of Theorem~\ref{thm:2submsts} to $T_L$ and $T_R$, similarly to the derivation of \eqref{eq:algbound}, obtain:
\begin{equation}
c_{mm}(ALG) ~\le~ 2\Big(c(T_L)+c(T_R)\Big) + (2w_\times-w_L-w_R).\label{eq:mmmst2}
\end{equation}
Combining \eqref{eq:mmmst1} and \eqref{eq:mmmst2}, we obtain:
\begin{eqnarray*}
{c_{mm}(ALG) \over c_{mm}(OPT_{mm})} &\le& {2\Big(c(T_L)+c(T_R)\Big) + (2w_\times-w_L-w_R) \over c(OPT)/2}\\
&\le& 2{ 2\Big(c(T_L)+c(T_R)\Big) + (2w_\times-w_L-w_R) \over \Big(c(T_L)+c(T_R)\Big) + (2w_\times - w_L - w_R) } \hskip .5cm \footnotesize{(\textit{applying } \eqref{eq:optbound})}\\
&\le& 4.\hskip 4.9cm \footnotesize{(\textit{because } w_\times \ge w_L,w_R)}
\end{eqnarray*}
%
%
\vskip -.55cm
\end{proof}

\fi

Extending Proposition \ref{obs:2submststight}, we obtain:
\begin{proposition}
There exist families of instances showing that the 2-MST min-sum and min-max approximation ratios are both tight.
\end{proposition}

\ifcccg
\else
\begin{proof}
Recall the graph constructed in the proof of Proposition~\ref{obs:2submststight}, and suppose that its $2n$ nodes consist of $n$ pairs, in each of which one node is located at $v_L$ (respectively, $v_R$) and the other is elsewhere in $\V_L$ (respectively, $\V_R$).

Now, first notice that the coloring defined in the proof of Proposition~\ref{obs:2submststight} is a valid red-blue coloring: the number of red nodes co-located at $v_L$ (respectively, $v_R$) equals the number of blue nodes located at other points of $\V_L$ (respectively, $\V_R$). Since $c(e_\times)=1+\epsilon$ and all other edges of MST $T$ are unit-weight, $\{v_L,v_R\}$ is the max-weight edge $e_\times$. Then in the resulting $V_L,V_R$, there are {\em no} lone nodes. Therefore the coloring of Proposition~\ref{obs:2submststight} is a coloring that could have been produced by Algorithm~\ref{alg:oumst}'s tie-breaking, justifying the names $ALG_b$ and $ALG_r$ for the two resulting colored trees.

Second, consider the following alternative coloring (see Fig.~\ref{fig:tightmstgood}): color all nodes at $v_L$ red (and thus all others in $V_L$ blue) and all nodes at $v_R$ blue (and thus all others in $V_R$ red), which is also a valid red-blue coloring, and also one that could have been produced by Algorithm~\ref{alg:oumst}'s tie-breaking. Call the two resulting colored trees $OPT_b$ and $OPT_r$. Then observe that $OPT_b$ (respectively, $OPT_r$) will pay once for $e_\times$ and for every edge of $T_L$ (respectively, $T_R$), and so 
\begin{equation}\label{eq:algbrbounds2}
c(OPT_b) = c(OPT_r) \approx c(T)/2.
\end{equation}
%
%
Combining \eqref{eq:algbrbounds} and \eqref{eq:algbrbounds2}, we conclude: ${c(ALG_b)+c(ALG_r) \over c(OPT_b)+c(OPT_r)} \to 3$ and ${\max\{c(ALG_b),c(ALG_r)\} \over \max\{c(OPT_b),c(OPT_r)\}} \to 2$.
\end{proof}
\fi


\subsection{Min-sum/min-max/bottleneck: hardness}

We provide a reduction inspired by the reduction of \cite{dyer1985complexity} from Three-Dimensional Matching to the problem of partitioning a bipartite graph into two connected components, each containing exactly half the vertices.

In our reduction, however, we reduce the traditional 3-SAT problem.



Given the 3-SAT formula, we construct the following graph (see Fig.~\ref{fig:MST}). For each clause, create a path of length $p$. For each variable $x_i$, we create create two nodes, $x_i$ and $\bar x_i$. We also create a path of length $p_b$ called $b$ and a path of length $p_r$ called $r$. From each $x_i$ or $\bar x_i$, we draw an edge to the final nodes of the paths corresponding the clauses that the literal appears in. Finally, from each $x_i$ {\em and} $\bar x_i$, we draw edges to the final nodes paths $b$ and $r$. All the edges defined have 1; all non-defined edges have weight 2. (In all cases when we refer to the ``final'' node of one of these $m+2$ paths, we mean the node with degree $>2$.)

\begin{figure}[t!]
\center
\ifcccg
\else
\includegraphics[width=5cm]{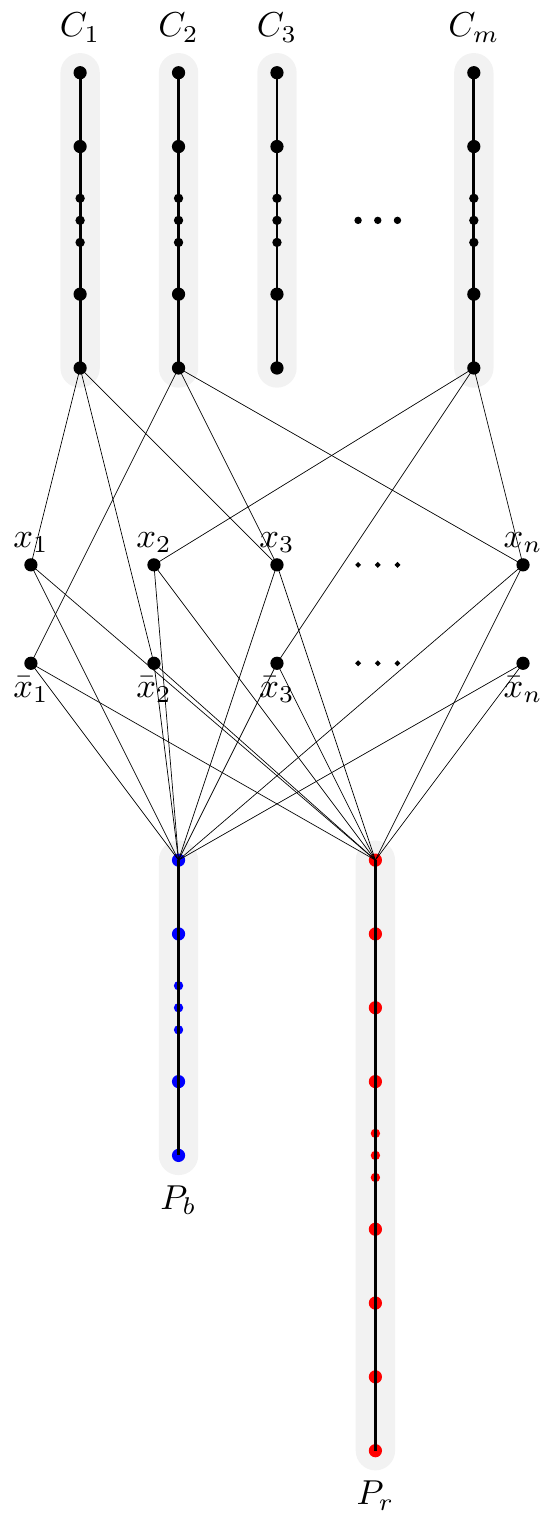}
\fi
\vskip -.1cm
\caption{Spanning tree reduction.}\label{fig:MST}
\end{figure}


The path lengths are defined as follows:
\ifcccg
\else
\begin{eqnarray*}
p_r &=& (n+1)\cdot n^3 + n +n +1\\
p_b &=& n^3 + n + 1\\
p &=& n^3 + 1.
\end{eqnarray*}
\fi


Then the total number of nodes in the graph constructed is:
\ifcccg
\else
\begin{eqnarray}
|V| &=& n \cdot  p + p_b + p_r + m\nonumber\\
&=& 2 \cdot (n_r + m).\label{eq:numnodes}
\end{eqnarray}
\fi

Finally, we must specify the $\{p_i,q_i\}$ pair relationships of these nodes. Each pair $\{x_i,\bar x_i\}$ is a $\{p_i,q_i\}$ pair. All $p_r$ nodes of path $p_r$ are $p_i$ s. All $p_b$ nodes of path $p$ and all $p$ nodes of path corresponding to an element are $q_i$ nodes. Observe that results in an equal number of $p_i$ and $q_i$ nodes since $p_b + n \cdot p = p_r$.

\begin{lemma}\label{lem:mstwt1iff}
The formula is satisfiable iff the constructed graph admits a 2-MST solution using only weight-1 edges.
\end{lemma}

\ifcccg
\else

\begin{proof}
First, suppose the formula admits a satisfying assignment. Then we color red the nodes of path $r$, node $x_i$ for each false $x_i$, and node $\bar x_i$ for each true $x_i$, and all other nodes blue. By (\ref{eq:numnodes}), this results in equal numbers of red and blue nodes without coloring both nodes of any pair the same. To obtain the resulting trees, we do the following: for each true $x_i$, delete edges $(x_i,r)$ and $(\bar x_i,b)$, and for each false $x_i$, delete edges $(x_i,b)$ and $(\bar x_i,r)$.

Second, suppose the graph admits a feasible solution. Suppose without loss of generality that the final node of path $r$ is colored red. In this case all $p_r$ nodes of path $r$ must be colored red. By (\ref{eq:numnodes}), exactly $m$ additional nodes must be colored red. Since $m < p_b$ and $m < p$, none of the nodes of path $p_b$ or of the element paths may be colored red, all of which must therefore be colored blue. This leaves $m$ blues and $m$ reds that must have been used to color the $x_i,\bar x_i$ nodes. In order for a clause path to have been colored blue, at least one of its three literals must have been colored blue. Moreover, since each pair of terminals $x_i,\bar x_i$ is a $\{p_i,q_i\}$ pair, we know they are colored different colors. Therefore we can read off a valid satisfying assignment from the colors of the literal nodes.
\end{proof}

\fi

Thus we conclude:

\begin{theorem}
In the special case of metric graphs with weights 1 and 2,  min-sum and min-max, 2-MST are both (strongly) NP-Complete, and bottleneck 2-MST is NP-hard to approximate with factor better than 2.
\end{theorem}

\section{2-TSP}

\subsection{Min-sum/min-max/bottleneck 2-TSP: hardness}

Clearly the min-sum and min-max objectives for 2-TSP are at least as hard to approximate as ordinary TSP in the same metric space (e.g., hard to approximate with factor better than 123/122 \cite{karpinski2015new}, even with edge weights $\{.5,1,1.5,2\}$): to reduce TSP to either of these, simply introduce a co-located pair $\{p_v,q_v\}$ for each node $v$ in the TSP instance.
\ifcccg
\else

\fi
Similarly, the same reduction implies that the bottleneck objective for 2-TSP is at least as hard to approximate as ordinary bottleneck TSP in the same metric space (e.g., hard to approximate with factor better than 2, even with edge weights $\{1,2\}$).


\subsection{Min-sum/min-max 2-TSP: algorithm}

Now we adapt Algorithm \ref{alg:oumst} above to obtain a \msTspFact-approximation algorithm for min-sum and min-max 2-TSP (see Algorithm \ref{alg:outsp}).

The proof again analyzes three cases, depending on whether one, both, or neither $T_L,T_R$ contains a pair.
Unlike with 2-MST, 2-TSP's $c(OPT)$ {\em is} lower-bounded by $c(T)$ in the first two cases, and so we can compare it to the simple upper bound on $c(ALG)$ of $4c(T)$.

\begin{algorithm}[t]
    \caption{Min-sum/min-max 2-TSP approx}
    \label{alg:outsp}
    \begin{algorithmic}[1] 
    \setcounter{ALG@line}{4}
    	\Statex \hskip -.575cm \textit{Identical to Alg,~\ref{alg:oumst}, except with lines 5,6 replaced by:}
%
	\State $C \gets$ a TSP tour, computed from $T$ by edge-doubling
	\ifcccg
\else
	            \For {$c \in \{b,r\}$}
	            	\Statex \hskip .52cm $C_c \gets$ a tour of the color-$c$ nodes, computed by shortcutting $C$
		\EndFor
	\fi
	\State \Return $\{C_b,C_r\}$
%
%
    \end{algorithmic}
\end{algorithm}


\begin{theorem}
Algorithm \ref{alg:outsp} is a \msTspFact-approximation algorithm for min-sum 2-TSP.
\end{theorem}
\ifcccg
\else
\begin{proof}
To upper-bound $c(OPT)$, we analyze three cases of the MST $T$, viz., where one, both, or neither of $\{T_L,T_R\}$ contain a pair.

\vskip .1cm
\noindent $\bullet$ (1) \textbf{Neither} $T_L$ nor $T_R$ contains a pair.
%
Then all nodes are lone nodes, i.e., $V_L$ are all blue and $V_R$ are all red. In this case, observe that $C_b$ will actually be the tour that would be obtained by edge-doubling $T_b$, and $C_r$ will be the tour that would be obtained by edge-doubling $T_r$. Thus in this case we have:
\begin{equation*}
c(ALG) ~\le~ c(C_b) + c(C_r) ~\le~ c(C) ~\le~ 2c(T) ~\le~ 2OPT.
\end{equation*}

\noindent $\bullet$ (2) \textbf{Only} (say) $T_R$ contains a pair, with all nodes in $V_L$ being (say) blue. In this case, all red nodes lie within $T_R$.

Since $T_R$ contains a pair, we know that $OPT$ must make at least one roundtrip between $V_L$ and $V_R$, costing at least $2w_\times$, and so $2w_\times \le c(OPT)$.

Now, consider the subgraphs $OPT_L,OPT_R$ of $OPT$ induced by $V_L,V_R$, respectively.

First, suppose $OPT_L$ is connected (in which case $OPT_R$ has exactly two components). Then $OPT_L$ consists of $|V_L|-1$ edges within $V_L \times V_L$ to connect $V_L$ together, and these edges must cost at least $c(T_L)$.

Since $OPT_R$ has two components, consisting of $|V_R|-2$ edges, they must cost at least $c(T_R^-)$.

Then combining the three contributions to the cost, we have $c(T_L) + 2w_\times + c(T_R^-) \le c(OPT)$, which implies:
\begin{equation}\label{eq:mstspoptlb1}
c(OPT) ~\ge~ c(T_L) + w_\times + c(T_R) ~=~ c(T).
\end{equation}

Second, suppose $OPT_L$ has exactly two components. Then it consists of $|V_L|-2$ edges within $V_L \times V_L$, costing at least $T_L^-$, but now also $OPT$ must make a second roundtrip between $V_L$ and $V_R$, costing at least $4w_\times$. The second visit to $V_R$ means that $OPT_R$ consists of $|V_R|-2$ edges, of total cost at least $c(V_R^-)-w_R$. That is, each additional component of $OPT_L$ reduces $c(OPT_L)$ by {\em at most $w_L$}, increases the cost due to cross edges by {\em at least $2w_\times$}, and decreases $c(OPT_R)$ by  {\em at most $w_R$}. Since $2w_\times \ge w_L+w_R$, the case of $OPT_L$ having multiple components would only increase the lower bound \eqref{eq:mstspoptlb1} further.


\vskip .1cm
\noindent $\bullet$ (3) \textbf{Both} $T_L$ and $T_R$ contain a pair. In this case, $OPT$ must make at least two roundtrips between $V_L$ and $V_R$, costing at least $4w_\times$, and so $4w_\times \le c(OPT)$.

First, suppose $OPT_L$ has exactly two components, one blue and one red (in which case $OPT_R$ also has exactly two components). Then $OPT_L$ consists of $|V_L|-2$ edges, and these edges must cost at least $c(T_L^-)$. Similarly, in this case $OPT_R$ consists of $|V_R|-2$ edges, and these edges must cost at least $c(T_R^-)$. Combining the three contributions, we have $c(T_L^-) + 4w_\times + c(T_R^-) \le OPT$, which implies:
\begin{equation}\label{eq:mstspoptlb2}
c(OPT) ~\ge~ c(T_L) + 2w_\times + c(T_R) ~\ge~ c(T).
\end{equation}

Second, suppose $OPT_L$ has an additional component, say, two blue and one red, which implies that $OPT_R$ also has two blue and one red, and that $OPT$ makes a third roundtrip between $OPT_L$ and $OPT_R$. The cost of $OPT_r$ is unchanged, but $OPT_b$ is increased by at least $2w_\times - w_L - w_R \ge 0$. More generally, therefore, the case of additional components would only increase lower bound \eqref{eq:mstspoptlb2} further.

Thus \eqref{eq:mstspoptlb1} holds in both cases (2) and (3).

Now we lower-bound $c(ALG)$ for these cases. 
Since the TSP tour $C$ is obtained from $T$ by edge-doubling, and then $C_b,C_r$ are both extracted from $C$ by shortcutting, we have:
\begin{equation}\label{eq:mstpalgup}
c(ALG) ~\le~ 2 c(C) ~\le~ 2 \cdot 2 c(T) ~=~ 4 c(T).
\end{equation}
Combining \eqref{eq:mstpalgup} and \eqref{eq:mstspoptlb1}, we conclude:
\begin{equation*}
{c(ALG) \over c(OPT)} ~\le~ 4.
\end{equation*}
\vskip -.6cm
\end{proof}
\fi

\begin{theorem}
Algorithm \ref{alg:outsp} is a \mmTspFact-approximation algorithm for min-max 2-TSP.
\end{theorem}
\ifcccg
\else
\begin{proof}
Let $OPT_{mm}$ and $c_{mm}(\cdot)$ be the optimal solution and cost function for min-max 2-TSP, respectively. Then $c_{mm}(OPT_{mm}) \ge c(OPT_{mm})/2 \ge c(OPT)/2$, and in particular, $c_{mm}(OPT_{mm}) \ge c(T)/2$.

Since the blue and red contributions to $c(ALG)$ were both upper-bounded by $2c(T)$ in all three of the ``neither,'' ``only,'' and ``both'' cases, we have $c_{mm}(ALG) \le 2c(T)$. Thus again we conclude:
\begin{equation*}
{c_{mm}(ALG) \over c_{mm}(OPT)} ~\le~ 4.
\end{equation*}
\vskip -.75cm
\end{proof}
\fi


\ifcccg
\else
\begin{figure*}[t!]\label{fig:tspstight}
\center
    \begin{subfigure}[t]{8.028388cm}
        \centering
            \includegraphics[width=7.5cm]{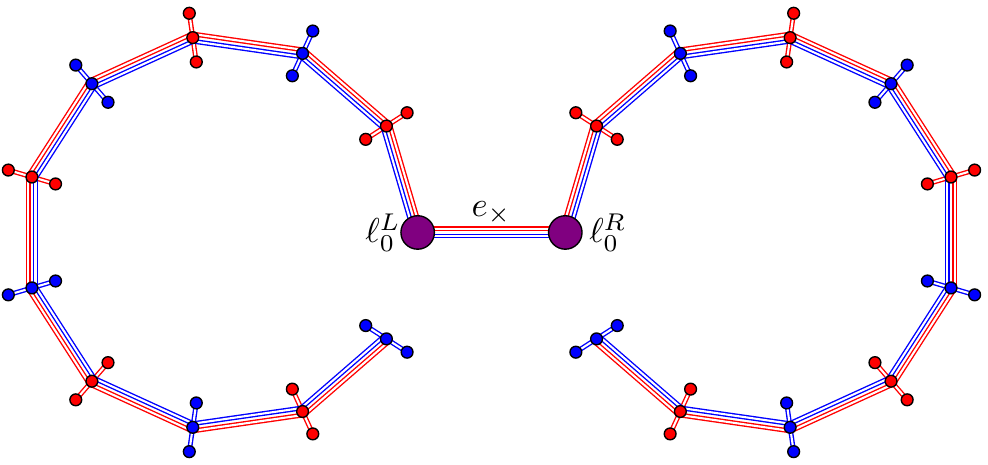}
\vskip .1cm
            \caption{Coloring with degree-4 nodes (and their leaves) alternating between red and blue (in which case half the nodes co-located at $\ell_0^L$ and $\ell_0^R$ are blue, and half red). Results (when tour $C$ visits just one of each degree-4 node's leaves before advancing to the next) in every edge paid for {\em twice} by $ALG_b$ and {\em twice} by $ALG_r$, and so $c(ALG_b) = c(ALG_r) \approx 2c(T)$.}
        \label{fig:tighttspbad}
    \end{subfigure}
\hskip .35cm
    \begin{subfigure}[t]{8.028388cm}
        \centering
            \includegraphics[width=7.5cm]{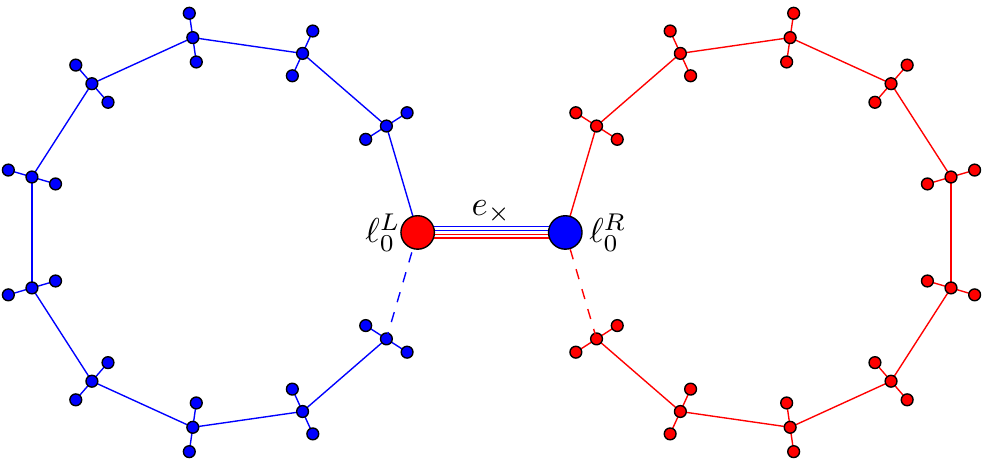}
\vskip .1cm
            \caption{Coloring with all nodes co-located at $\ell_0^L$ red and all those at $\ell_0^R$ blue (in which case all others in $T_L$ are blue and all others in $T_R$ are red). Results (when tour $C$ immediately visits both of each degree-4 node's leaves upon reaching it) in $e_\times$ paid for once by $OPT_b$ and once by $OPT_r$, and all other edges of $T$ (plus the two edges shown dashed) paid for once total, and so $c(OPT_b) = c(OPT_r)\approx c(T)/2$.}
        \label{fig:tighttspgood}
    \end{subfigure}
\caption{
2-TSP instance achieving Algorithm~\ref{alg:outsp}'s approximation factor 4 for min-sum and min-max, drawn with two colorings. Its $2n=12\lambda$ {\em nodes} are {\em (co-)located} at the $2\n=2(3\lambda+1)$ {\em points} shown. $3\lambda$ nodes are co-located at each of $\ell_0^L$ and $\ell_0^R$, and one node is located at every other point. Each node pair has one node at $\ell_0^L$ (or $\ell_0^R$) and one node at another point of $T_L$ (respectively, $T_R$).}
%
\label{fig:tighttspgoodbad}
\end{figure*}
\fi

\begin{proposition}\label{obs:tspstight}
There exist families of instances
showing that the 2-TSP min-sum and min-max approximation factors are both tight.
\end{proposition}
\ifcccg
\else
\begin{proof} 
We construct a graph as follows. First consider a set of $2\n = 2(3\lambda+1)$ points $\V = \V_L \cup \V_R$ in a metric space (with $|\V_L|=|\V_R|=\n$, and $\lambda$ even), arranged in the form of two {\em caterpillar trees} $\T_L,\T_R$, i.e., graphs with the property that removal of all leaves results in a path graph (see Fig.~\ref{fig:tighttspgoodbad}). In particular, for each of $s \in \{L,R\}$, let the path contained within $\T_s$ follow the points $\ell_0^s,...,\ell_\lambda^s$, where $c(\{\ell_{j-1}^s,\ell_j^s\})=1$ for all $j \in [\lambda]$, and $c(\{\ell_\lambda^s,\ell_0^s\})=1+\epsilon$. Also, for each $j \in [\n]$, and for both $s \in \{L,R\}$, let there be two leaf points, each distance $\epsilon$ from $\ell_j^s$. Finally, let $c(\{\ell_0^L,\ell_0^R\})=1+\epsilon$. Let $\T = \T_L \cup \{\{\ell_0^L,\ell_0^R\}\} \cup \T_R$, and let the distances between all other pairs of points of $\V$ equal their path distances in $\T$.

Now we define a metric-weighted graph on $2n=12\lambda$ nodes $V$, which are located at points of $\V$ as follows. First, for $s \in \{L,R\}$ and $j \in [\lambda]$, one node each is located at $\ell_j^s$ and at each of $\ell_j^s$'s two leaf neighbors. Second, for $s \in \{L,R\}$, $3\lambda$ nodes are co-located at $\ell_0^s$. Each node co-located at $\ell_0^s$ is the partner of a node already placed at another point of $\V_s$.

Now, observe that the max-weight edge $e_\times$ of an MST $T$ on $V$ will be $\{\ell_0^L,\ell_0^R\}$ (or more precisely, an edge between a node located at $\ell_0^L$ a node located at $\ell_0^R$); the resulting $T_L,T_R$ will (if we ignore weight-0 edges between nodes co-located at the same point) both be caterpillar trees; and there will be {\em no} lone nodes in the resulting $V_L,V_R$. Thus every valid coloring assigning distinct colors to each pair's nodes is a can potentially be chosen by Algorithm~\ref{alg:outsp}'s tie-breaking. Because of the degree-4 nodes of $T_L,T_R$, tie-breaking will also play a significant role in the computation of $C,C_b,C_r$ by edge-doubling.

Now, consider two following two colorings and edge-doubling computations.

%
%
%

First (see Fig.~\ref{fig:tighttspbad}), for both $s \in \{L,R\}$, let $\ell_j^s$ and its two leaves be colored blue for all odd $j \in [\lambda]$, and red for all even $j \in [\lambda]$, in which case half the nodes co-located at $\ell_0^s$ are red and half are blue. Moreover, suppose that when $C$ is constructed through edge-doubling (starting from, say, $\ell_0^L$), the tie-breaking determining the order of edges traversed is done in such a way that upon (the first) arrival at each degree-4 node $\ell_j^s$, {\em only one} of its two leaves is visited before advancing to $\ell_{j+1}^s$, with the result that after eventually reaching $\ell_\lambda^s$, the tour visits both its leaves and then doubles back, visiting $\ell_{\lambda-1}^s$'s second leaf, $\ell_{\lambda-2}^s$'s second leaf, and so on, circumnavigating $T_L$ a second time in reverse.
Then $C$ crosses {\em both copies} of each doubled edge of $T$, with no shortcutting savings, and so $c(C)=2c(T)$. Because red and blue alternate back and forth about all of $C$, the $C_b,C_r$ extracted from $C$---call them $ALG_b,ALG_r$---will also obtain no shortcutting savings, each costing the same as $C$, and so
\begin{equation}\label{eq:tspbrbounds2}
c(ALG_b)=c(ALG_r)=2c(T).
\end{equation}


Second (see Fig.~\ref{fig:tighttspgood}), let all nodes co-located at $\ell_0^L$ be colored red (thus all the other nodes of $V_L$ are blue), and all nodes co-located at $\ell_0^R$ be blue (thus all the other nodes of $V_L$ are red). Moreover, suppose that when tour $C$ is constructed (starting from, say, $\ell_0^L$), the tie-breaking is such that upon arrival at each degree-4 node $\ell_j^s$ (for $s \in \{L,R\}$), its two leaves are visited before advancing to $\ell_{j+1}^s$. After visiting $\ell_\lambda^L$ and its leaves, therefore, $C$ will shortcut to $\ell_\lambda^R$, costing only $c(\{\ell_\lambda^L,\ell_0^L\})+c(e_\times)$.
That is, $C$ will pay twice for $e_\times$ and for the $\epsilon$-weight leaf edges, and it will also once for $\{\ell_\lambda^L,\ell_0^L\}$ and $\{\ell_\lambda^R,\ell_0^R\}$ (shown dashed in Fig.~\ref{fig:tighttspgood} because they are not edges of $T$), but it only pays {\em once} for each non-leaf edge of $T_L$ and $T_R$. Moreover, the $C_b$ extracted from $C$---call it $OPT_b$---will, after visiting $\ell_0^R$, shortcut past the rest of $T_R$, returning directly to $\ell_0^L$. That is, it will pay twice for $e_\times$ and $T_L$'s leaf edges, and will also pay once for $\{\ell_\lambda^L,\ell_0^L\}$, but it will only pay once each for $T_L$'s non-leaf edges.
The behavior of $C_r$---call it $OPT_r$---will be symmetric, and so
\begin{equation}\label{eq:tspbrbounds}
c(OPT_b) = c(OPT_r) \approx c(T)/2.
\end{equation}



Combining \eqref{eq:tspbrbounds2} and \eqref{eq:tspbrbounds}, we conclude: ${c(ALG_b)+c(ALG_r) \over c(OPT_b)+c(OPT_r)} \to 4$ and ${\max\{c(ALG_b),c(ALG_r)\} \over \max\{c(OPT_b),c(OPT_r)\}} \to 4$.
\end{proof}



\fi

\section{2-Matching}

\subsection{Preliminaries}


In the case of perfect matching we require that the number of pairs $n$ be even. It will be convenient to re-express the 2-Matching problem as an equivalent problem concerning cycle covers.

We begin with some observations about the nature of feasible solutions in this setting. By definition, two nodes $p_i,q_i$ from the same pair can never be matched because they must receive different colors. Each must then be matched with a node of the same color, and each of {\em those} nodes's partners must receive the opposite color and be matched with a node of that color, and so on, in a consistent fashion. One way to make this consistency requirement concrete is the following alternative description. First, for each pair $\{p_i,q_i\}$, draw a length-2 path (of unit-weight edges) between them, separated by a dummy node $d_i$, and in the resulting $3n$-node graph $G'$ consider instead the task of finding a 2-factor, i.e., a node-disjoint cycle cover, of minimum cost. In particular, consider seeking a cycle cover that uses only unit-weight edges, which would have cost $3n$.

\begin{definition}
Say that a 2-matching or cycle cover is {\em feasible} if it uses only unit-weight edges. We call a non-dummy node of $G'$ (i.e., a node from $G$) a {\em real node}; similarly, we call an edge between a dummy node and a real node $G'$ a {\em dummy edge} and a path $p_i d_i q_i$ a {\em dummy path}; we call an edge between two real nodes a {\em real edge}.
\end{definition}

Observe that any feasible 2-matching in $G$ will induce a 2-factor of $G'$: imagine $G'$ drawn in a ``tripartite'' style, with the red nodes in the left column, the blue nodes in the right column, and the dummy nodes in the center column. Then each path $p_i - d_i - q_i$ forms a ``cross-edge'' (going either left or right), each red edge appears in the left column, and each blue edge appears in the right column. Each non-dummy node is matched with one other node in the 2-matching, so if we combine the edges of the paths $p_i - d_i - q_i$ to those of the matching, then in the graph induced by these edges, each of the $3n$ nodes will have degree 2. This implies the edge set is a 2-factor. Note that the cost of the 2-factor differs by a known amount ($2n$, because each dummy nodes two edges are unit-weight)) from the (min-sum) cost of the corresponding 2-matching.

The problem of finding a minimum-cost 2-factor is known to be polynomial-time solvable by reduction to bipartite  matching (folklore). Unfortunately, a 2-factor of $G'$ will not necessarily induce a valid 2-matching on $G$. In $G'$ as defined, the additional property needed (somewhat analogously to bipartite graphs having no odd cycles) is for {\em each cycle's size to be a multiple of 6}, which we will call a \Ccnsp.

\begin{definition}
Let a \Cc for a given graph be a 2-factor, i.e., a node-disjoint collection of subgraphs covering all nodes, where each subgrraph is a member of $\{C_6, C_{12}, C_{18}, ...\}$.
\end{definition}

\begin{lemma}\label{lem:conn2covs}
Any feasible \Cc for $G'$ will induce a feasible 2-matching for $G$.
\end{lemma}
\ifcccg
\else
\begin{proof}
Each dummy node $d_i$ has degree-2, with edges to $p_i$ and $q_i$, and every non-dummy node has exactly one dummy neighbor. Therefore in any feasible \Ccnsp, each dummy node $d_i$'s two edges $\{p_i,d_i\},\{q_i,d_i\}$ must appear; moreover, for each real node, exactly {\em one} of its real edges must appear in the \Ccnsp. Thus every such cycle must alternate between single real edges and length-2 dummy paths.

Given the cycle cover, we can therefore construct a valid 2-coloring consistent with the matching it induces by performing the following procedure on each cycle appearing in the cycle. Choose one of its real nodes (say, $p_i$) arbitrarily, and color it (say) red. Then color its dummy neighbor's other neighbor $q_i$ blue, and also color $q_i$'s real neighbor (say, $p_{i'}$) blue. Then go to $p_{i'}$'s dummy neighbor's other neighbor (say, $p_{i''}$), and check whether $p_{i''}$ is the starting node $p_i$. If not, color it red and repeat. Since the roundtrip from $p_i$ back to $p_i$ must involve crossing an even number of dummy paths, it will never happen that we inconsistently attempt to color $p_i$ blue when we return to it.
\end{proof}
\fi

Unfortunately, unlike the problem of deciding whether a graph admits a feasible cycle cover, deciding whether it admits a \Cc is NP-Complete \cite{hell1984packings}. This fact does not immediately imply the hardness of the 2-Matching problems, however, because $G'$ is not an arbitrary graph. We can characterize it as follows. It contains $3n$ nodes consisting of $n$ triples $\{p_i,d_i,q_i\}$, where each $d_i$ is degree 2, with neighbors $p_i,q_i$.

\begin{figure*}[t!]
\center
    \begin{subfigure}[t]{7.5cm}
        \centering
            \includegraphics[width=6.75cm]{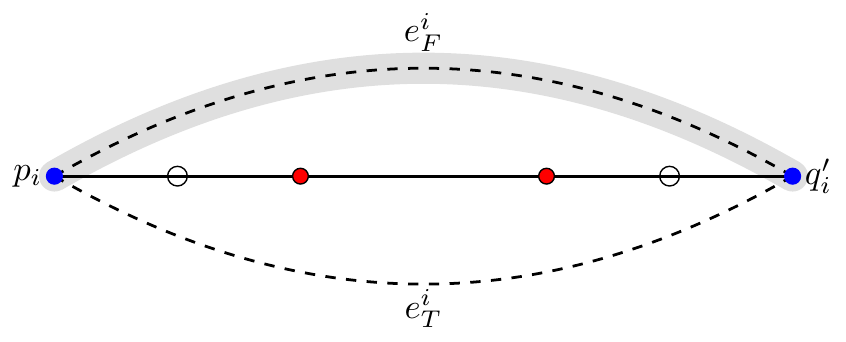}
\ifcccg
\else
\vskip -.05cm
\fi
            \caption{Variable gadget for $x_i$. Any feasible \Cc must include pseudoedge $e_F^i$ xor edge $e_T^i$.}
        \label{fig:var}
    \end{subfigure}
\hskip .5cm
    \begin{subfigure}[t]{7.5cm}
        \centering
            \includegraphics[width=5.5cm,clip=true,trim={0 0 0 .45cm}]{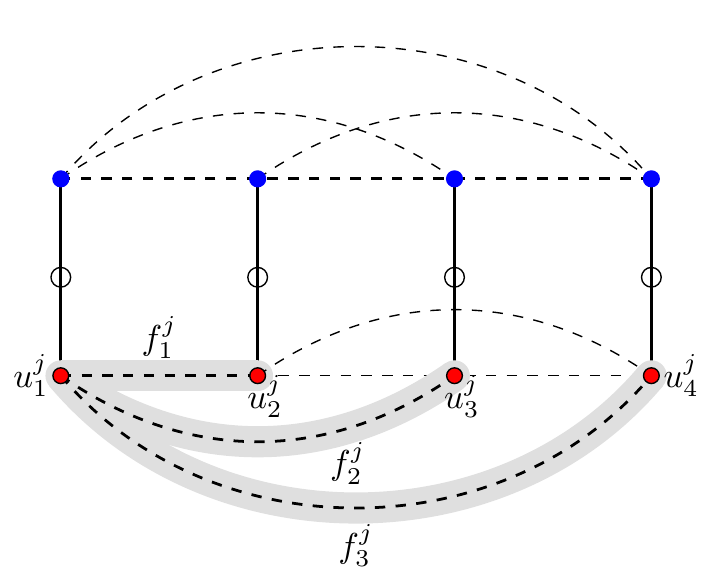}
\ifcccg
\else
\vskip -.05cm
\fi
            \caption{Clause gadget for $C_j$. Any feasible \Cc must include \textit{exactly one} of the three distinguished pseudoedges $f_1,f_2,f_3$ (plus one of the unlabeled dashed edges from the bottom and two of the top).}
        \label{fig:clause}
    \end{subfigure}

\ifcccg
\else
\vskip .27cm
\fi

    \begin{subfigure}[b]{\textwidth}
        \centering
            \includegraphics[width=\textwidth]{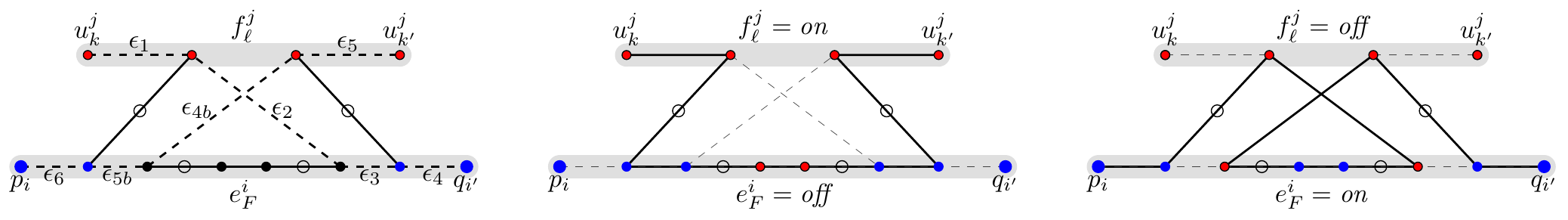}
\ifcccg
\else
\vskip -.05cm
\fi
            \caption{Connection gadget (left fig.) for an appearance (negated iff $v=F$) of variable $x_i$ in clause $C_j$. The lower shaded path is a more detailed view of one of the $x_i$ gadget's pseudoedge $e_F^i$ (see (\subref{fig:var})); the upper shaded path is a more detailed view of one of the $C_j$ gadget's three distinguished pseudoedges $f_1^j,f_2^j,f_3^j$ (see (\subref{fig:clause})).
We show (see Lemma~\ref{lem:conn2covs}) that there are only two possible feasible {\Ccnsp}s of the gadget, one in which $f_\ell^j$ is \textit{on} and $e_F^i$ is \textit{off}, meaning this connection represents $C_j$'s unique true literal (middle fig.), and one in which $f_\ell^j$ is \textit{off} and $e_F^i$ is \textit{on}, meaning it represents one of $C_j$'s two false literals (right fig.).}
        \label{fig:consistent}
    \end{subfigure}

\vskip .15cm


    \begin{subfigure}[b]{\textwidth}
        \centering
            \includegraphics[width=\textwidth]{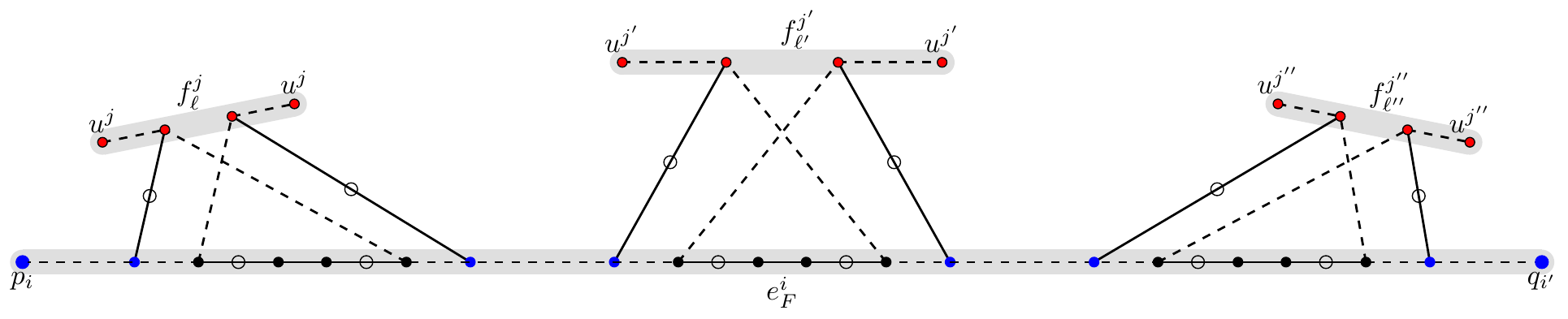}
\ifcccg
\else
\vskip -.05cm
\fi
            \caption{The example $e_F^i$ shown here is a more detailed view of pseudoedge $e_F^i$ in variable $x_i$'s gadget (see (\subref{fig:var})).
An $e_F^i$ can have multiple connections ({\em in this example}, three), corresponding to appearances of $x_i$'s in different clauses (in this example, an $x_i$ literal appears as the $\ell$th literal in clause $C_j$, and so on; typically $j,j',j''$ will {\em all be distinct}); an $f_\ell^j$ has only one connection, since it indicates what literal the $\ell$th literal in clause $C_j$ is. Subscripts of $u$ nodes are omitted for clarity.}
        \label{fig:nestedconns}
    \end{subfigure}
    
\caption{Gadgets used in 2-Matching's hardness proof. Real nodes are shown filled in, dummy nodes unshaded. Edges that must be used in any feasible solution are shown solid, other edges dashed. $e_F^i$ and $f_1^j,f_2^j,f_3^j$ are pseudoedges, i.e., schematic representations of paths that connections attach to.}\label{fig:matchgadgets}
\end{figure*}

\subsection{Bottleneck 2-Matching: hardness}

To prove hardness, we give a reduction inspired by Papadimitriou's reduction \cite{cornuejols1980matching} from 3-SAT to the problem of deciding whether a graph can be partitioned into a node-disjoint collection of cycles, each of size {\em at least 6}.

We reduce from \textsc{Monotone 1-in-3 SAT} (which has no negated literals) to the problem of deciding whether $G'$ admits a (feasible, i.e., using unit-weight edges only) \Ccnsp.  Recall that edge weights in $G'$ are 1 or 2, and that each dummy node's two edges are weight-1. Given the boolean formula, we proceed as follows. 

For each variable $x_i$, we create a gadget as shown in Fig.~\ref{fig:var}. It consists of a 6-path $(p_i,d_i,q_i,p'_i,d'_i,q'_i)$, whose nodes form two triples $\{p_i,d_i,q_i\},\{p'_i,d'_i,q'_i\}$, plus an edge $(p_i,q'_i)$ labeled $e_T^i$ and a {\em pseudoedge} labeled $e_F^i$. There will be exactly two feasible ways to cover the nodes of this gadget in a \Ccnsp, with the cycle including $e_i^T$, corresponding to $x_i$ being true, and the one including $e_i^F$, corresponding to false.

For each clause $C_j$, we create a gadget as shown in Fig.~\ref{fig:clause}. It consists of two copies of $K_4$, where each node $u_\ell^j$ in one $K_r$ is connected by a 2-path and dummy node to a corresponding node $v_\ell^j$ in the other. Three {\em pseudoedges} connecting a distinguished node $u_0^j$ to the other three nodes of the same $K_4$ are labeled $f_1^j,f_2^j,f_3^j$. If a feasible \Ccnsp, one of these edges will be on and the other two off, corresponding to a satisfied \textsc{1-in-3 SAT} clause.

\begin{definition}
A {\em pseudoedge} is an edge, or the result of attaching a connection gadget to a pseudoedge.
\end{definition}

Finally, to implement the appearance of a variable in a clause, we use the gadget shown in Fig.~\ref{fig:consistent}, which will appear in sequence.
Applying a connection gadget to pseudoedges $e_F^i$ and $f_\ell^j$ does the following:
\begin{enumerate}
\item the last (rightmost) edge of $f_\ell^j$ is split into a 9-path path via the creation of 8 new nodes (compare $e_F^i$ in Figs.~\ref{fig:var}, \ref{fig:consistent}(left), and \ref{fig:nestedconns});
\item $f_\ell^j$'s edge is replaced with two new edges (labeled $\epsilon_1,\epsilon_5$ in Fig.~\ref{fig:consistent}) incident to two new nodes (compare $f_\ell^j$ in Figs.~\ref{fig:clause} and \ref{fig:consistent}(left));
\item $f_\ell^j$'s first new node is connected to $e_F^i$'s first and seventh new nodes, by a 2-path and an edge, respectively (see Fig.~\ref{fig:consistent}(left)); and
\item $f_\ell^j$'s second new node is connected to $e_F^i$'s second and eighth new nodes, by an edge and a 2-path, respectively (see Fig.~\ref{fig:consistent}(left)).
\end{enumerate}


For each variable $x_{i}$ appearing (in some position $k \in [3]$) within a clause $C_j$, we draw a connection gadget between $x_i$'s $e_F^{i}$ and $C_j$'s $f_k^j$. 
%
%
%
%
First observe the following, which can be verified by inspection:

\begin{fact}
If all pseudoedges $e_F^i$ and $f_k^j$ were simply edges, then a \Cc would induce one of two legal states within any variable $X_i$'s gadget, with exactly one of $e_F^i,e_T^i$ on, and one of three legal states within any clause $C_j$'s gadget, with exactly one of $f_1^j,f_2^j,f_3^j$ on.
\end{fact}

Now we show that any \Cc will induce one of two canonical states on each connection gadget (see Fig.~\ref{fig:consistent} middle and right), each pseudoedge, and each variable gadget.

\begin{lemma}\label{lem:connstate}
Within any pseudoedge pair $(e_F^i,f_k^j)$
connected by a connection gadget, a feasible \Cc induces one of only two legal states: one with the first and last edges (labeled $\epsilon_1$ and $\epsilon_5$, respectively, in Fig.~\ref{fig:consistent}(left)) within $f_k^j$ on (``$f_k^j$ is on''), and the other with with the first and last edges (labeled $\epsilon_6$ and $\epsilon_4$, respectively, in Fig.~\ref{fig:consistent}(left)) within $e_F^i$ on (``$e_F^i$ is on'').
\end{lemma}

\ifcccg
\else

\begin{proof}
First, assume $f_k^j$ is $e_F^i$'s only connection.
Suppose edge $\epsilon_1$ is on (see Fig.~\ref{fig:consistent}(left)).
$\epsilon_1$ on implies $\epsilon_2$ off, which implies $\epsilon_3$ on, which implies both $\epsilon_4$ and (because otherwise a 9-cycle would be formed) $\epsilon_{4b}$ off; $\epsilon_{4b}$ off implies both $\epsilon_5$ and $\epsilon_{5b}$ on; and $\epsilon_{5b}$ implies $\epsilon_6$ off. Similarly, if instead $\epsilon_6$ is on, then this will eventually imply that $\epsilon_4$ is on and that both $\epsilon_1$ and $\epsilon_5$ are off.

Now suppose $f_k^j$ is only one of multiple connections of $e_F^i$'s, say, the first (leftmost) one (see Fig.~\ref{fig:nestedconns}). But the first connection's $\epsilon_4$ edge (see Fig.~\ref{fig:consistent}(left)) is also the second connection's $\epsilon_6$ edge. Therefore by repeated application of the single-connection argument, the result follows for the general case.
\end{proof}

\fi

This immediately implies:

\begin{corollary}
A feasible \Cc induces one of two canonical states within each variable gadget and one of three canonical states within each cause gadget.
\end{corollary}


In a solution where the clause's edge $f^j_k$ is on, this forces $e_{i_k}^F$ to be off, and hence $e_{i_k}^T$ to be on; similarly, it forces clause $C_j$'s other two distinguished pseudoedges to be off, and hence the variables connected to those edges to be false. (The clause gadget's other edges can be freely used or not, as needed to form a feasible \Ccnsp.)
\ifcccg
\else



\fi
Finally, observe that the final constructed graph $G'$ indeed satisfies the required structure for corresponding to an equivalent instance $G$ of the 2-Matching problem: every dummy node has exactly two neighbors (both real), and every real node has exactly one dummy neighbor. 

From the arguments above, we conclude that $G'$ admits an all-unit weight \Cc iff $G$ admits an all-unit weight 2-matching iff the underlying boolean formula is satisfiable. 
Thus we conclude:

\begin{theorem}
In the special case of metric graphs with weights 1 and 2, bottleneck 2-Matching is NP-hard to approximate with factor better than 2 (and min-sum and min-max 2-Matching are both (strongly) NP-Complete).
\end{theorem}

\ifcccg
\else
\begin{proof}
For min-max, observe that the two resulting matchings will use all unit edges iff the formula is satisfiable.

For bottleneck's hardness of approximation, observe that any solution will be forced to use some weight-2 (i.e., nonexistent) edge iff the formula is unsatisfiable.
\end{proof}
\fi

\subsection{Min-sum/min-max 2-Matching: hardness}

By reduction from a special case of \textsc{Max 1-in-3 SAT}, we can obtain a hardness of approximation result for the min-sum and min-max objectives. Let \textsc{Max 1-in-3 SAT-5} denote \textsc{Max 1-in-3 SAT} under the restriction that each variable appears in at most 5 clauses.

Lampis has shown (implicitly in \cite{lampis2012improved}\footnote{Karpinksi et al.~\cite{karpinski2015new} provide a similar construction yielding a stronger hardness of approximation lower bound for Metric TSP, but adapting that construction to our present problem actually leads to a slightly weaker lower bound.}) the following:

\begin{lemma}
There exists a family of \textsc{Max 1-in-3 SAT-5} instances with $15m$ clauses and $8.4m$ variables, each appearing in at most 5 clauses, for which, for any $\epsilon>0$, 
it is NP-hard to decide whether the minimum number of unsatisfiable clauses is at most $\epsilon m$ or at least $(0.5-\epsilon)m$.
%
%
\end{lemma}

For concreteness, let \textsc{Min Not-1-in-3 SAT-5} indicate the optimization problem of minimizing the number of unsatisfied clauses in a \textsc{1-in-3 SAT-5} formula.

%

%


Now we argue that the same construction used above provides an approximation-preserving reduction from \textsc{Min Not-1-in-3 SAT-5}.

\begin{corollary}\label{cor:summaxmatchhard}
Min-sum and min-max 2-Matching are both, in the special case of metric graphs with weights 1 and 2, NP-hard to approximate with factor better than $8305/8304 \approx 1.00012$.
\end{corollary}

\ifcccg
\else

\begin{proof}(Sketch.)
By inspection of the construction's gadgets, we observe that given a satisfying assignment for the \textsc{Min Not-1-in-3 SAT-5} formula, the corresponding matching problem solution will use 16 weight-1 edges per clause, 6 per variable, and 12 per connection. This yields a total min-sum cost of exactly 
$15m \cdot 16 + 8.4m \cdot 6 + 15m \cdot 3 \cdot 12 = 830.4m$.


Now suppose the formula's optimal solution leaves $k$ clauses unsatisfied. 
Unsatisfied clauses will force the resulting matching problem solution to use weight-2 edges, either within the clause gadget or elsewhere. Because variables are limited to 5 appearances, 
$k$ unsatisfied clauses will necessitate the use of at least $k/5$ weight-2 edges, each replacing a weight-1 edge, thus increasing the solution cost by at least $k/5$. $(0.5-\epsilon)m$ unsatisfied clauses imply an added cost of $(0.5-\epsilon)m/5 \approx 0.1m$.

Hence the two specified types of \textsc{Min Not-1-in-3 SAT-5} instances that are NP-hard to distinguish will translate into matching problem instances with optimal min-sum solution costs approximately $830.4m$ and $830.5m$, respectively.

Similarly, for a satisfiable formula, the matching problem instance will have an optimal min-max cost of exactly $830.4m/2 = 415.2m$. $(0.5-\epsilon)m$ unsatisfied clauses imply that the maximum of the two resulting tree weights will increase by at least approximately $0.1m/2 = 0.05m$. This leads to min-max solution costs approximately $415.2m$ and $415.25m$, respectively, again yielding the same ratio.
%
%
%
%
\end{proof}

\fi

\ifcccg
\else
\section*{Acknowledgements}
\fi
This work was supported in part by NSF award INSPIRE-1547205, and by the Sloan Foundation via a CUNY Junior Faculty Research Award.
We thank Ali Assapour, Ou Liu, and Elahe Vahdani for useful discussions.

\bibliographystyle{abbrv}
\bibliography{bib}

\end{document}